\documentclass[a4paper]{amsproc}

\allowdisplaybreaks[2]

\newcommand{\RR}{\mathbb{R}}
\newcommand{\one}{\mathbf{1}_{\RR_+}}
\newcommand{\cO}{\mathcal{O}}
\DeclareMathOperator{\sgn}{sgn}

\newtheorem{theorem}{Theorem}

\title[Eigenvalues of star graphs]{Variational proof of the existence of~eigenvalues for star graphs}

\author{Konstantin Pankrashkin}

\email{konstantin.pankrashkin@math.u-psud.fr}
\urladdr{http://www.math.u-psud.fr/~pankrash/}

\address{Laboratoire de Math\'ematiques d'Orsay, Univ.~Paris-Sud, CNRS, Universit\'e\\ Paris-Saclay, 91405 Orsay, France}

\dedicatory{Dedicated to Pavel Exner on the occasion of his 70th birthday}

\begin{document}

\begin{abstract}
We provide a purely variational proof of the existence
of eigenvalues below the bottom of the essential spectrum
for the Schr\"odinger operator with an attractive $\delta$-potential
supported by a star graph, i.e. by a finite union of rays emanating
from the same point. In contrast to the previous works, the construction is valid without
any additional assumption on the number or the relative position
of the rays. The approach is used to obtain an upper bound for the lowest eigenvalue.
\end{abstract}

\maketitle

\section{Introduction}

The mathematically rigorous study of multidimensional
Schr\"odinger operators with potentials supported by hypersurfaces
was initiated in 1994 by Brasche, Exner, Kuperin and \v Seba \cite{BEKS}.
The two-dimensional Hamiltonians with interactions supported
by curves have become a prominent class of solvable models
of quantum mechanics \cite{ex16} and are usually referred to as leaky quantum graphs.
A summary of various questions and results in the spectral theory
of such operators can be found in the review by Exner \cite{ex-aga},
and for the most recent developments we refer to the papers \cite{bel,dr,ek,ep,ev,kl,lotor}
and to Chapter~10 in the recent monograph by Exner and  Kova\v r\'\i k~\cite{ekbook}.

In the present contribution, we are interested in some properties
of Schr\"odin\-ger operators with $\delta$-interactions supported by
the so-called star graphs.
By a star graph $\Gamma$ we mean a subset of $\RR^2$ obtained as the union of finitely many rays
emanating from the origin. If $(r,\theta)$ is the standard
polar coordinate system, then $\Gamma$ is naturally identified
with a family $(\theta_1,\dots,\theta_N)$ in which $0\le \theta_1<\dots<\theta_N<2\pi$
by $\Gamma:=\bigcup_{j=1}^N \big\{(r,\theta):\theta=\theta_j,\, r\ge 0\big\}$.
The associated Schr\"odinger operator $H_{\Gamma,\alpha}=-\Delta-\alpha \delta_\Gamma$, where $\delta_\Gamma$
is the Dirac $\delta$-distribution supported by $\Gamma$
and $\alpha>0$ is a coupling constant, is defined as the unique
self-adjoint operator in $L^2(\RR^2)$ associated with the closed lower semibounded quadratic form
\[
Q_{\Gamma,\alpha}(u)=\iint_{\RR^2} |\nabla u|^2dx - \alpha \int_\Gamma |u|^2 ds,
\quad u\in H^1(\RR^2).
\]
where $ds$ is the one-dimensional Hausdorff measure, cf.~\cite{BEKS}.
Such configurations appear naturally as a mathematical model for a junction
of quantum wires, and they were first analyzed by Exner and N\v emcov\'a~\cite{en1,en2}.
The basic spectral properties of the operator are well known:
the essential spectrum coincides with the semi-axis
$[-\alpha^2/4,+\infty)$, and the discrete spectrum is non-empty except in the degenerate cases when $\Gamma$ is a single ray ($N=1$) or a straight line ($N=2$ and $|\theta_1-\theta_2|=\pi$).
Despite the simple geometrical picture, the only available proof of the existence of eigenvalues
is based on a rather involved analysis of integral operators carried out by Exner and Ichinose~\cite{ei}. On the other hand, by the min-max principle,
the non-emptiness of the discrete spectrum would follow from the existence
of a trial function $v\in H^1(\RR^2)$ satisfying the strict inequality
\begin{equation}
  \label{eq-1}
Q_{\Gamma,\alpha}(v)< -\dfrac{\alpha^2}{4}\|v\|^2_{L^2(\RR^2)}.
\end{equation}
Surprisingly, the construction of such a function appeared to be a difficult task.
The construction of Exner and N\v emcov\'a~\cite{en2} works only if there
is a pair of rays with $|\theta_j-\theta_k|\mod 2\pi <0.092$. Brown, Eastham and Wood \cite{bew1,bew2,eastham}
managed to find a trial function for all possible configurations with $N\ge 3$
as well as for the configurations with $N=2$ and $|\theta_1-\theta_2|<0.9271$.
In the present note we show how to construct such a function
for all possible cases (Theorem~\ref{thm1}), and our approach uses
a likeliness between the star graphs and a spectral problem of the surface superconductivity
with a similar geometry
discussed by Lu and Pan~\cite{lupan} and Helffer and Morame~\cite{hm}.
We remark again that Theorem~\ref{thm1} itself does not provide any new spectral information,
but suggests a new method to show the presence of a non-empty discrete spectrum
as an alternative to the analytical proof by Exner and Ichinose~\cite{ei}. 
On the other hand, the presence of explicitly given
trial functions allows one to obtain a universal upper
bound for the lowest eigenvalue (Theorem~\ref{thm2}), which is a new result.

\section{Construction of a trial function}

By the min-max principle, it is sufficient to consider the case $N=2$ (a broken line),
then, up to isometries, all possible configurations can be described
by a single parameter $\theta\in(0,\pi/2)$
through $\Gamma=\overline{\mathstrut\Gamma_+\cup\Gamma_-}$ with $\Gamma_\pm:=\big\{(t,\pm t \tan\theta):t>0\big\}$,
and the associated operator $H_{\Gamma,\alpha}$ will be denoted by $H(\theta,\alpha)$.

We remark first that in order to show that the discrete spectrum is non-empty
it is sufficient to consider the problem in the half-plane $\RR\times\RR_+$, i.e. to find a function $u\in H^1(\RR\times\RR_+)$ satisfying
\[
\iint_{\RR\times\RR_+}|\nabla u|^2dx -\alpha\int_{\Gamma_+} |u|^2dx < -\dfrac{\alpha^2}{4}\,
\|u\|^2_{L^2(\RR\times\RR_+)},
\]
as its extension $v$ to the whole of $\RR^2$ by parity 
automatically satisfies \eqref{eq-1}. For subsequent constructions,
it is handy to perform an additional rotation to put the support
of the interaction onto the positive semi-axis of ordinates.
In other words, we will work with the domain
$\Omega:=\big\{(x_1,x_2): x_1< x_2 \tan \theta \big\}$
and the quadratic form
\[
Q(u)=\iint_{\Omega} |\nabla u|^2 dx - \alpha\int_{\RR_+} \big|u(0,x_2)\big|^2dx_2,
\quad u\in H^1(\Omega).
\]

\begin{theorem}\label{thm1}
Pick any $\rho\in (0,\cot^2\theta)$ and any Lipschitz function $\chi:\RR\to[0,1]$ with $\chi(t)=1$
for $|t|\le 1$ and $\chi(t)=0$ for $|t|\ge 2$, then
for sufficiently large $n>0$ the function $u$ defined by
\begin{equation}
 \label{eq-u}
u(x_1,x_2)=e^{-\alpha|x_1|/2}
\bigg(
\dfrac{2}{\alpha} \,\one(x_2) -\dfrac{1}{\alpha}\,e^{-\alpha|x_2|\tan\theta}\sgn x_2
\bigg)^\rho \chi\Big(\dfrac{x_2}{n}\Big)
\end{equation}
satisfies the strict inequality $Q(u)<-\dfrac{\alpha^2}{4}\, \|u\|^2_{L^2(\Omega)}$.
\end{theorem}

\begin{proof}
For futher use, denote
\[
F(t):=\int_{-\infty}^t e^{-\alpha|x_1|} \,dx_1= \dfrac{2}{\alpha} \,\one(t) -\dfrac{1}{\alpha}e^{-\alpha|t|}\sgn t.
\]
For the functions $u$ of the form $u(x_1,x_2)=e^{-\alpha |x_1|/2} g(x_2)$ with real-valued $g$ we have
\begin{equation}
   \label{eq-a}
\|u\|^2_{L^2(\Omega)}=\int_\RR \int_{-\infty}^{x_2\tan\theta}e^{-\alpha|x_1|}g(x_2)^2dx_1dx_2=
\int_\RR g(x_2)^2 F(x_2\tan\theta)\,dx_2.
\end{equation}
Furthermore, 
\begin{multline*}
Q(u)=\dfrac{\alpha^2}{4}\int_\RR g(x_2)^2 \int_{-\infty}^{x_2\tan\theta}e^{-\alpha|x_1|}\,dx_1\,dx_2\\
+\int_\RR g'(x_2)^2 \int_{-\infty}^{x_2\tan\theta}e^{-\alpha|x_1|}dx_1dx_2
-\alpha \int_{\RR_+}g(x_2)^2dx_2.
\end{multline*}
Due to
\begin{gather*}
\dfrac{\alpha^2}{4}\int_{-\infty}^{x_2 \tan\theta} e^{-\alpha|x_1|}dx_1 =\dfrac{\alpha^2}{4}\,F(x_2\tan \theta)=-\dfrac{\alpha^2}{4}\,F(x_2\tan \theta)
+\dfrac{\alpha^2}{2}\,F(x_2\tan \theta)\\
=-\dfrac{\alpha^2}{4}\,F(x_2\tan\theta)+\alpha \,\one(x_2)-\dfrac{\alpha}{2}\,e^{-\alpha|x_2|\tan\theta}\sgn x_2
\end{gather*}
we have
\begin{multline}
    \label{eq-b}
Q(u)=-\dfrac{\alpha^2}{4} \int_\RR g(x_2)^2F(x_2\tan \theta)\,dx_2+
\int_\RR g'(x_2)^2 F(x_2\tan \theta)\,dx_2\\
 -\dfrac{\alpha}{2}\int_\RR g(x_2)^2 e^{-\alpha|x_2|\tan\theta}\sgn x_2 \,dx_2.
\end{multline}
Using the integration by parts we obtain
\begin{equation}
   \label{eq-c}
\begin{gathered}
\int_\RR g(x_2)^2 e^{-\alpha|x_2|\tan\theta}\sgn x_2 \,dx_2
=\dfrac{2 }{\alpha}\,\cot\theta \int_\RR g(x_2)g'(x_2)e^{-\alpha|x_2|\tan\theta}dx_2\\
=\dfrac{2 }{\alpha}\,\cot\theta \int_\RR g(x_2)g'(x_2)F'(x_2\tan\theta)\,dx_2,
\end{gathered}
\end{equation}
and the substitution of \eqref{eq-a} and \eqref{eq-c} into \eqref{eq-b} gives the representation
\begin{align*}
Q(u)&=-\dfrac{\alpha^2}{4}\|u\|^2_{L^2(\Omega)}+R(g),\\
R(g)&:=\int_{\RR}g'(x_2) \Big( g'(x_2)F(x_2\tan\theta)- g(x_2)F'(x_2\tan\theta)\cot \theta\Big)\,dx_2.
\end{align*}
Hence, we need to find a function $g$ with $R(g)<0$.

Pick $\rho\in (0,\cot^2\theta)$ and introduce a function $g_\rho$ by
$g_\rho(x_2)=F(x_2\tan\theta)^\rho$, then
\[
R(g_\rho)=\rho \tan^2\theta ( \rho-\cot^2\theta)
\int_\RR e^{-2\alpha |x_2|\tan\theta} F(x_2\tan\theta)^{2\rho-1}dx_2<0.
\]
Remark that the integral is finite, but the function $g_\rho$
has a non-zero finite limit at $+\infty$, and the associated function
$u$ does not belong to $H^1(\Omega)$ due to~\eqref{eq-a}.

Choose a Lipschitz function $\chi:\RR\to[0,1]$ with $\chi(t)=1$ for $|t|\le 1$ and $\chi(t)=0$ for $|t|\ge 2$,
and for $n>0$ denote $h_n:=g_\rho \chi(\cdot/n)$. By construction,
the associated functions $u_n$ given by
\begin{equation}
     \label {eq-un}
u_n(x_1,x_2)=e^{-\alpha|x_1|/2}h_n(x_2),
\end{equation}
belong to $H^1(\Omega)$ and coincide with \eqref{eq-u}.
In addition,
\begin{gather*}
R(h_n)-R(g_\rho)\\
=\int_{\RR} \Big(\chi\Big(\dfrac{x_2}{n}\Big)^2-1\Big)g'_\rho(x_2)\Big( g'_\rho(x_2)F(x_2\tan\theta)- g_\rho(x_2)F'(x_2\tan\theta)\cot \theta\Big)dx_2\\
+ \dfrac{1}{n}\int_{\RR} \chi\Big(\dfrac{x_2}{n}\Big) \chi'\Big(\dfrac{x_2}{n}\Big)\Big( 2g_\rho(x_2)g'_\rho(x_2)F(x_2\tan\theta)- g_\rho(x_2)^2F'(x_2\tan\theta)\cot \theta\Big)dx_2\\
+ \dfrac{1}{n^2}\int_{\RR} \chi'\Big(\dfrac{x_2}{n}\Big)^2g_\rho(x_2)^2 F(x_2\tan\theta)\,dx_2=:I_1+I_2+I_3.
\end{gather*}
Due to the finiteness of $R(g_\rho)$, for large $n$ we have
\begin{multline*}
|I_1|\le \int_{\RR\setminus(-n,n)}\Big|g'_\rho(x_2)\Big( g'_\rho(x_2)F(x_2\tan\theta)- g_\rho(x_2)F'(x_2\tan\theta)\cot \theta\Big)\Big|dx_2\\
=\rho \tan^2\theta \cdot \big| \rho-\cot^2\theta\big|\int_{\RR\setminus(-n,n)} e^{-2\alpha |x_2|\tan\theta} F(x_2\tan\theta)^{2\rho-1}dx_2=o(1).
\end{multline*}
Furthermore, 
\begin{align*}
|I_2| &=\bigg|\dfrac{1}{n}\int_{\RR}\big(2\rho \tan\theta -\cot\theta\big) \chi\Big(\dfrac{x_2}{n}\Big)
\chi'\Big(\dfrac{x_2}{n}\Big)e^{-\alpha|x_2|\tan\theta}F(x_2\tan\theta)^{2\rho}\,dx_2\bigg|\\
&\le\dfrac{|2\rho -\cot^2\theta|\cdot \tan\theta\cdot\|\chi'\|_\infty}{n}
\int_{\RR} e^{-\alpha|x_2|\tan\theta}F(x_2\tan\theta)^{2\rho}\,dx_2 =\cO\Big(\dfrac{1}{n}\Big)
\end{align*}
due to the convergence of the integral.
Finally, as the integrand is bounded, we have
\[
|I_3|\le \dfrac{1}{n^2} \bigg(\int_{-2n}^{-n} +\int_{n}^{2n} \|\chi'\|_\infty^2\, g_\rho(x_2)^2F(x_2\tan\theta)\,dx_2\bigg)= \dfrac{1}{n^2} \cdot \cO(n)
=\cO\Big(\dfrac{1}{n}\Big),
\]
and we arrive at $R(h_n)=R(g_\rho)+o(1)$ as $n$ tends to $+\infty$.
As $R(g_\rho)<0$, we have $R(h_n)<0$ for large $n$,
which shows that the functions \eqref{eq-un} have the sought property.
\end{proof}

\section{Upper bound for the lowest eigenvalue}

We remark first that various estimates for the lowest eigenvalue $\lambda(\theta,\alpha)$
of $H(\theta,\alpha)$ were obtained in earlier works. In particular, Duch\^ene and Raymond \cite{dr} showed
that
\begin{equation}
  \label{eq-dr}
\lambda(\theta,\alpha)=-\alpha^2\Big[ 1-c_1\theta^{2/3}+\cO(\theta)\Big], \quad \theta\to 0+,
\end{equation}
and Exner and Kondej \cite{ek} proved that
\begin{equation}
  \label{eq-ek}
\lambda(\theta,\alpha)=-\alpha^2\Big[ \dfrac{1}{4} + c_2 \big(\dfrac{\pi}{2}-\theta\big)^4+o\Big(\big(\dfrac{\pi}{2}-\theta\big)^4\Big)\Big],
\quad \theta\to \dfrac{\pi}{2}\,-,
\end{equation}
where $c_1$ and $c_2$ are some explicit positive constants.

Recall that by the min-max principle there holds
$\lambda(\theta,\alpha)\le Q(v)/\|v\|^2_{L^2(\Omega)}$
for any non-zero $v\in H^1(\Omega)$. We would like to use
the trial functions $u$ from Theorem~\ref{thm1}
to obtain an explicit upper estimate for the eigenvalue
valid for all values of~$\theta$. As the limit
$\lim_{n\to+\infty }Q(u)/\|u\|^2_{L^2(\Omega)}=-\alpha^2/4$ coincides with the bottom of the essential
spectrum, we cannot hope for the best possible result.
Nevertheless, the estimate  and the method can be of some interest as,
to our best knowledge, no analogous bound has been available so far.

\begin{theorem}\label{thm2}
For any $\theta\in\Big(0,\dfrac{\pi}{2}\Big)$ there holds
$\lambda(\theta,\alpha)\le -\alpha^2\Big( \dfrac{1}{4} + \Lambda(\theta)\Big)$,
where
\begin{gather}
    \label{eq-lambda}
\Lambda(\theta):=
\dfrac{3\cos^6\theta\big(2^{2\cos^2\theta}-1\big)^2}{2\,\big(1+2\cos^2\theta\big)^3\big(108+180 \cos^2\theta-132\cos^4\theta +45 \cos^6\theta - 5 \cos^8\theta\big)}
\end{gather}
is strictly positive.
\end{theorem}

A comparison with \eqref{eq-dr} and \eqref{eq-ek} shows that the upper estimate is away of an optimal one.
For $\theta$ close to $0$ our estimate gives $\lambda(\theta,\alpha)\le -99\alpha^2/392+\cO(\theta)$
which is very weak when compared with the true behavior given by \eqref{eq-dr}.
At $\theta=\pi/2$, the value of $\Lambda(\theta)$ vanishes at the tenth order, which is also very
far from the true fourth order given in \eqref{eq-ek}.  Our bound resulted from various experiments
with the parameters and used a number of very rough inequalities, and the interested reader should
feel free to improve the estimate using an alternative choice of parameters.

\begin{proof}
The result is based on a more accurate estimate
of the quantities appearing in the proof of Theorem~\ref{thm1} for
an explicit choice of the function $\chi$ and of the parameter $\rho$.
Namely, we set
\[
\chi(t):=\begin{cases}
1, & |t|\le 1, \\
2-|t|, & |t|\in(1,2),\\
0, & |t|\ge 2,
\end{cases}
\qquad
\rho:=\cos^2\theta,
\]
then $\|\chi'\|_\infty=1$. We have
\begin{multline*}
R(g_\rho)=\dfrac{\rho \tan^2\theta ( \rho-\cot^2\theta)}{\alpha^{2\rho-1}}\\
\times \bigg(
\int_{-\infty}^0 e^{(2\rho+1)x_2\tan\theta}dx_2+ \int_0^{+\infty} e^{-2\alpha x_2\tan\theta}\big(2-e^{-\alpha x_2\tan\theta}\big)^{2\rho-1}dx_2
\bigg).
\end{multline*}
We calculate
\[
\int_{-\infty}^0 e^{(2\rho+1)x_2\tan\theta}dx_2= \dfrac{1}{(2\rho+1)\alpha \tan\theta}
\]
and, using the change of variables $s=e^{-\alpha x_2\tan\theta}$,
\begin{multline*}
\int_0^{+\infty} e^{-2\alpha x_2\tan\theta}\big(2-e^{-\alpha x_2\tan\theta}\big)^{2\rho-1}dx_2=\dfrac{1}{\alpha \tan\theta}\int_0^1 s(2-s)^{2\rho-1}ds\\
=\dfrac{1}{\alpha \tan\theta} 
\int_0^1 \Big(2(2-s)^{2\rho-1}-(2-s)^{2\rho} \Big)\,ds
=\dfrac{1}{\alpha \tan\theta} \Big(\,\dfrac{2^{2\rho}-1}{\rho} - \dfrac{2^{2\rho+1}-1}{2\rho+1}\Big),
\end{multline*}
which gives
\[
R(g_\rho)=\dfrac{1}{\alpha^{2\rho}} \dfrac{\tan\theta\,(\rho-\cot^2\theta)(2^{2\rho}-1)}{2\rho+1}=
-\dfrac{\cos^3\theta(2^{2\cos^2\theta}-1)}{\sin\theta (1+2\cos^2\theta)}.
\]

In what follows we will use the following estimates valid for $s\in[0,1]$
due to the convexity argument:
\[
\dfrac{1}{1+2s}\le1-\dfrac{2}{3}\, s, \quad
\dfrac{1}{(1+2s)^2}\le1-\dfrac{8}{9}\, s, \quad
2^{2s}\le 1+3s.
\]
We estimate
\begin{align*}
|I_1|&\le \dfrac{\cos^4\theta}{\alpha^{2\rho-1}} \bigg(
\int_{-\infty}^{-n}
e^{(2\rho+1)\alpha x_2\tan\theta}dx_2\\
&\qquad
+
\int_{n}^{+\infty} e^{-2\alpha x_2\tan\theta} \big(2-e^{-\alpha x_2\tan\theta}\,\big)^{2\rho-1}dx_2
\bigg)\\
&\le 
\dfrac{\cos^4\theta}{\alpha^{2\rho-1}} \,\bigg(
\int_{-\infty}^{-n}
e^{(2\rho+1)\alpha x_2\tan\theta}dx_2
+
2^{2\rho} \int_{n}^{+\infty} e^{-2\alpha x_2\tan\theta}dx_2
\bigg)\\
&=\dfrac{\cos^4\theta}{\alpha^{2\rho-1}} \,\bigg(
\dfrac{1}{(2\rho+1)\alpha \tan \theta} e^{-(2\rho+1)\alpha n\tan\theta}
+
\dfrac{2^{2\rho}}{2 \alpha \tan\theta} e^{-2\alpha n\tan\theta}
\bigg)\\
&\le
\dfrac{\cos^4\theta}{\alpha^{2\rho-1}} \,\bigg(
\dfrac{1}{\big((2\rho+1)\alpha \tan \theta\big)^2n}
+
\dfrac{2^{2\rho}}{(2 \alpha \tan\theta)^2 n}
\bigg)\\
&\le
\dfrac{1}{\alpha^{2\rho+1}}\cdot
\dfrac{\cos^6\theta}{\sin^2\theta} \,\Big(
\dfrac{1}{(2\rho+1)^2}
+
\dfrac{1}{4}\cdot 2^{2\rho}\Big)\cdot \dfrac{1}{n}\\
&\le
\dfrac{1}{\alpha^{2\rho+1}}\cdot
\dfrac{\cos^6\theta}{\sin^2\theta} \,\Big(
1-\dfrac{8}{9}\,\rho
+
\dfrac{1}{4} \,(1+3\rho)\Big)\cdot \dfrac{1}{n}\\
&=
\dfrac{1}{\alpha^{2\rho+1}}\cdot \dfrac{45 \cos^6\theta - 5 \cos^8\theta}{36 \sin^2\theta}\cdot \dfrac{1}{n}
\end{align*}
and
\begin{align*}
|I_2|&\le \dfrac{\big|2\cos^2 \theta -\cot^2\theta\big|\cdot\tan\theta}{n\alpha^{2\rho}}\Big(\int_{-\infty}^{\,0} e^{(2\rho+1)\alpha x_2\tan\theta}dx_2
+2^{2\rho}\int_0^\infty e^{-\alpha x_2\tan\theta}dx_2\Big)\\
&=\dfrac{|2\sin^2 -1|\cdot \cos\theta}{n\alpha^{2\rho}\sin\theta} \Big( \dfrac{1}{(2\rho+1)\alpha\tan\theta}
+\dfrac{2^{2\rho}}{\alpha \tan\theta}
\Big)\\
&=\dfrac{1}{\alpha^{2p+1}}\dfrac{\cos^2\theta \cdot \big|2\sin^2\theta-1\big|}{\sin^2\theta}
\Big(\,
\dfrac{1}{2\cos^2\theta+1}+2^{2\cos^2\theta}
\Big)\\
&\le 
\dfrac{1}{\alpha^{2p+1}}\dfrac{\cos^2\theta}{\sin^2\theta}
\Big(
1 -\dfrac{2}{3}\,\cos^2\theta + 1 +3\cos^2\theta
\Big)\\
&=\dfrac{1}{\alpha^{2p+1}}\cdot \dfrac{72 \cos^2\theta+ 84\cos^4\theta}{36 \sin^2\theta} \cdot \dfrac{1}{n} .
\end{align*}

Finally, the bounds $|F|\le 1/\alpha$ on $\RR_-$ and $|F|\le 2/\alpha$ on $\RR_+$ give
\begin{align*}
|I_3|&\le \dfrac{1}{n^2} \bigg(\int_{-2n}^{-n} +\int_{n}^{2n} g_\rho(x_2)^2F(x_2\tan\theta)\,dx_2\bigg)\\
&\le
\dfrac{1}{n^2}\bigg( \Big(\dfrac{1}{\alpha}\Big)^{2\rho+1}n + \Big(\dfrac{2}{\alpha}\Big)^{2\rho+1}n\bigg)
= \dfrac{1}{\alpha^{2\rho+1}} \Big(1+2^{2\rho+1}\Big) \cdot \dfrac{1}{n}\\
&\le \dfrac{1}{\alpha^{2\rho+1}} \Big(1+2 \big(1+3\rho\big)\Big) \cdot \dfrac{1}{n}
=\dfrac{1}{\alpha^{2\rho+1}} \cdot \Big(\,3+6\cos^2\theta\,\Big) \cdot \dfrac{1}{n}\\
&=\dfrac{1}{\alpha^{2\rho+1}} \cdot \dfrac{108 \sin^2\theta + 216 \sin^2\theta\cos^2\theta}{36 \sin^2\theta} \cdot \dfrac{1}{n}.
\end{align*}
As a result, we obtain
\begin{multline*}
R(h_n)\le R(g_\rho)+|I_1|+|I_2| + |I_3|\le -\Big(a - \dfrac{b}{n}\Big),\\
\begin{aligned}
a&:=-R(g_\rho), \quad b:= \dfrac{1}{\alpha^{2\rho+1}}\cdot \dfrac{B}{36 \sin^2\theta},\\[\smallskipamount]
B&:=108 \sin^2\theta + 72 \cos^2\theta+ \big(84\cos^2\theta +216 \sin^2\theta\big)\cos^2\theta+45 \cos^6\theta - 5 \cos^8\theta\\
&\phantom{:}=108-36\cos^2\theta + \big(216-132\cos^2\theta\big)\cos^2\theta+45 \cos^6\theta - 5 \cos^8\theta\\
&\phantom{:}=108+180 \cos^2\theta-132\cos^4\theta +45 \cos^6\theta - 5 \cos^8\theta,
\end{aligned}
\end{multline*}
implying
\[
Q(u)+ \dfrac{\alpha^2}{4}\,\|u\|^2_{L^2(\Omega)}\le R(h_n)\le -\Big(a-\dfrac{b}{n}\Big).
\]
On the other hand,
\begin{align*}
\|u\|^2_{L^2(\Omega)}&\le \int_{-2n}^{2n} g_\rho(x_2)^2 F(x_2\tan\theta)\,dx_2\\
&= \int_{-2n}^{0} F(x_2\tan\theta)^{2\rho+1}dx_2
+ \int_{0}^{2n} F(x_2\tan\theta)^{2\rho+1}dx_2\\
&\le
2n\Big( \dfrac{1}{\alpha}\Big)^{2\rho+1}+2n\Big( \dfrac{2}{\alpha}\Big)^{2\rho+1}
\le \dfrac{1}{\alpha^{2\rho+1}} \cdot \Big(2+ 4\cdot 2^{2\rho}\Big)\cdot n\le cn\\
\text{with} \quad c&:=\dfrac{6(1+2\cos^2\theta)}{\alpha^{2p+1}},
\end{align*}
and we have
\[
\mu(\theta,\alpha):=-\dfrac{\alpha^2}{4}-\lambda(\theta,\alpha)
\ge  \dfrac{an-b}{c n^2} \quad
\text{provided } \quad an >b.
\]
The right-hand side is optimized by $n=2b/a$ resulting in
$\mu(\theta,\alpha)\ge \dfrac{a^2}{4 bc}=\alpha^2 \Lambda(\theta)$
with $\Lambda(\theta)$ given in \eqref{eq-lambda}.
\end{proof}

\end{document}